\documentclass[11pt,letterpaper]{article}
\usepackage{graphicx,amssymb,amsmath,url,amsthm}
\usepackage[dvipsnames]{xcolor}
\usepackage{fullpage}
\usepackage[font=small]{caption}
\usepackage{lineno}
\usepackage{hyperref}
\usepackage[inline]{enumitem}

\newtheorem{lemma}{Lemma}
\newtheorem{corollary}{Corollary}
\newtheorem{theorem}{Theorem}
\newtheorem{observation}{Observation}

\newcommand{\area}{{\rm area}}
\newcommand{\eps}{\varepsilon}

\title{Maximum Area Axis-Aligned Square Packings\thanks{Research supported in part by the NSF awards CCF-1422311 and CCF-1423615. The first author was supported by the Science Without Borders program.}}

\date{}
\author{
	Hugo A. Akitaya\thanks{Tufts University, Medford, MA, USA.
		\protect\url{hugo.alves_akitaya@tufts.edu},\protect\url{matthew.jones@tufts.edu}}
	\and
	Matthew D. Jones\footnotemark[2]
	\and
	David Stalfa~\thanks{Department of Computer Science, Northeastern University, Boston, MA, USA. \protect\url{stalfa@ccis.neu.edu}}
	\and
	Csaba D. T\'oth\footnotemark[2]~\thanks{California State University Northridge, Los Angeles, CA, USA. \protect\url{csaba.toth@csun.edu}}
}

\index{Akitaya, Hugo}

\begin{document}
\nolinenumbers
\maketitle

\begin{abstract}
Given a point set $S=\{s_1,\ldots , s_n\}$ in the unit square $U=[0,1]^2$, an \emph{anchored square packing} is a set of $n$ interior-disjoint empty squares in $U$ such that $s_i$ is a corner of the $i$th square. The \emph{reach} $R(S)$ of $S$ is the set of points that may be covered by such a packing, that is, the union of all empty squares anchored at points in $S$.

It is shown that $\area(R(S))\geq \frac12$ for every finite set $S\subset U$, and this bound is the best possible. The region $R(S)$ can be computed in $O(n\log n)$ time. Finally, we prove that finding a maximum area anchored square packing is NP-complete. This is the first hardness proof for a geometric packing problem where the size of geometric objects in the packing is unrestricted.
\end{abstract}

\section{Introduction}
\label{sec:intro}

Let $S=\{s_1,\ldots , s_n\}$ be a set of $n$ points in the unit square $U=[0,1]^2$. We say that a square $q$ is \emph{empty} if no point in $S$ lies in the interior of $q$, and $q$ is \emph{anchored} at a point $s$ if one of its four corners is $s$. An \emph{anchored square packing for $S$} is a set $Q=\{q_1,\ldots , q_n\}$ of  interior-disjoint axis-aligned empty squares that lie in $U$ such that $q_i$ is anchored at $s_i$ for $i=1,\ldots , n$. A \emph{lower-left} anchored square packing is an anchored square packing in which $s_i$ is the lower-left corner of $q_i$, for $i=1,\ldots , n$~\cite{BT16}. No polynomial-time algorithm is known for computing the maximum area of an anchored square packing for a given point set $S$; the problem admits a PTAS using a reduction to the maximum weight independent set problem (MWIS)~\cite{BDT17}. The empty squares anchored at $S$ do not always cover $U$ entirely (Fig.~\ref{fig:upperBound}(a)). For finding a maximum anchored square packing for $S$, it suffices to consider the subset of $U$ that can be reached by anchored empty squares. Specifically, we define the \emph{reach} of $S$, denoted $R(S)$, as the union of all axis-aligned empty squares contained in $U$ and anchored at some point in $S$.

For computing the reach $R(S)$, we can take the union of all maximal empty squares anchored at the points in $S$, as follows. For $i=1,\ldots ,n$, let $q_i^1$ be the maximal axis-aligned empty square in $U$ whose lower-left corner is $s_i$, and similarly define $q_i^2$, $q_i^3$, and $q_i^4$ where $s_i$ is the upper-left, upper-right, and lower-right corner, respectively.
We say that a point $s \in S$ \emph{blocks} a square $q_i^1$ if $s$ is incident to the top or right edge of $q_i^1$.
Similarly, $s$ blocks $q_i^j$ if $j=2$ (resp., 3, 4) and $s$ is incident to the bottom or right edges  of $q_i^j$ (resp.,  bottom or left edges, or top or left edges of $q_i^j$).
It is now clear that $R(S)=\bigcup_{i=1}^n \bigcup_{j=1}^4 q_i^j$.

\begin{figure}[h]
	\centering
\includegraphics[width=.9\textwidth]{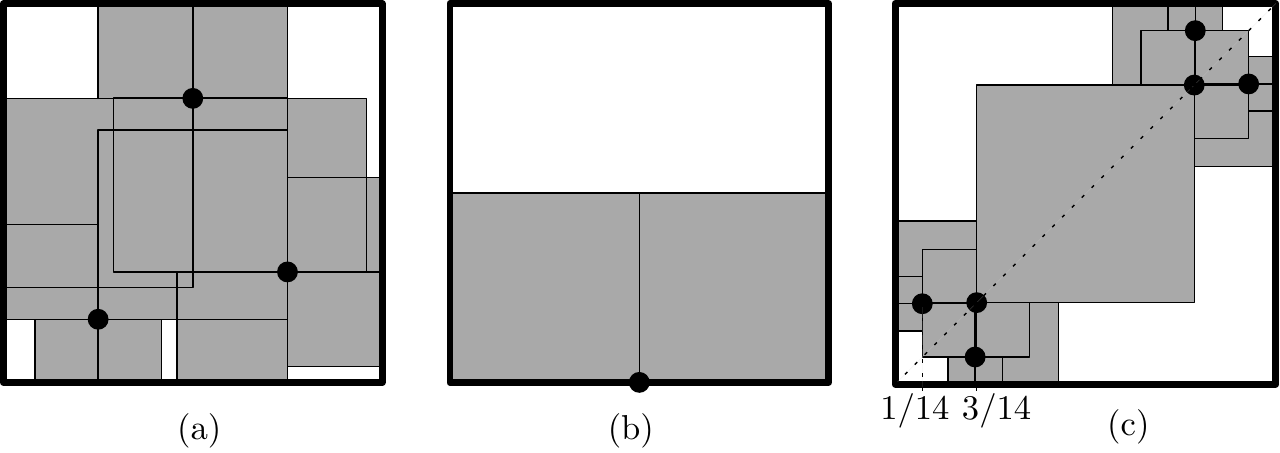}
	\caption{(a) The reach $R(S)$ for a set $S$ of three points.
             (b) The area of $R(S)$ is $\frac12$ for $S=\{(\frac12,0)\}$.
             (c) The reach $R(S)$ touches all four sides of $U$, and its area is $\frac{4}{7}$.}
	\label{fig:upperBound}
\end{figure}

\subparagraph{Summary of Results.}
We prove that for every finite set $S\subset U$, the area of $R(S)$ is at least $\frac{1}{2}$, and this bound is the best possible (Section~\ref{sec:lowerbound}). This settles in the affirmative a conjecture by Balas et al.~\cite{BDT17}. We show how to compute $R(S)$ in $O(n\log n)$ time where $n=|S|$ (Section~\ref{sec:algorithm}). We also show that finding the maximum area anchored square packing for a given point set $S$ is NP-complete (Section~\ref{sec:hardness}). This is the first NP-hardness result for a geometric packing problem, where the
size of the geometric objects in the packing is unrestricted.
We conclude with related open problems (Section~\ref{sec:open}).

\subparagraph{Motivation and Related Previous Work.}
Geometric packing and covering problems have a long and revered history, going back to Kepler's problem about the densest packing of congruent balls in Euclidean space. In a classical packing problem, we are given a container region $C$, and a set $O$ of geometric objects, and we wish to find a maximum subset $O'\subseteq O$ such that congruent copies (or translates) of the objects in $O'$ fit in $C$ without overlap.

Anchored variants, where each geometric object needs to contain a given point (\emph{anchor}) initially emerged in VLSI design, where the anchors represent the endpoints of wires. Allen Freedman~\cite{Tu69} conjectured that for every finite set $S\subset [0,1]^2$, which contains the origin (i.e., $\mathbf{0}\in S)$, there is a lower-left anchored rectangle packing of area at least $\frac12$. This lower bound would match an easy upper bound construction, where $n$ points are equally distributed on the diagonal. The current best lower bound is $0.091$~\cite{DT15}.

More recently, a broad family of anchored packing problems were proposed in the context of \emph{map labeling}, where the anchors represent cities in a map, and axis-aligned rectangles represent labels~\cite{FW91,IL03,JC04,KT13,KR92,KNN+02,KSW99}. Variants of the problem require the anchor to be at a corner, at a side, or anywhere in the rectangle, and the objective is to maximize the number of labels that can be packed in the map. Many of these problems are known to be NP-complete. However, in all previous reductions, the label boxes have a finite number of possible sizes~\cite{FW91,KR92,KSW99} or bounded size~\cite{IL03}.

In this paper, we consider the variant of Freedman's problem: We need to place an axis-aligned square at each anchor, and the sizes of the squares are not given in advance. Our objective is to maximize the total area of an anchored square packing. Balas et al.~\cite{BDT17} showed that a greedy strategy finds an $\frac{5}{32}$-approximation, and a reduction to MWIS yields a PTAS that achieves an $(1-\varepsilon)$-approximation in time $n^{O(1/\varepsilon)}$. It is known that the number of maximum-area square anchored packings may be exponential in $n$~\cite{BT16}.

\section{The Minimum Area of the Reach}
\label{sec:lowerbound}

In this section, we prove $\area(R(S))\geq \frac{1}{2}$ for every set $S$ of $n$ points in $U=[0,1]^2$ (Theorem~\ref{thm:area}). Note that this bound is the best possible for all $n\in \mathbb{N}$. Indeed, if $S$ is the one-element set $S=\{(\frac{1}{2},0)\}$, then $\area(R(S))=\frac{1}{2}$; see Fig.~\ref{fig:upperBound}(b). By placing $n$ points in an $\varepsilon$-neighborhood of $(\frac{1}{2},0)$ in $U$, we see that for every $\varepsilon>0$ and every $n\in \mathbb{N}$, there exists a set $S$ of $n$ points in $U$ such that $\area(R(S))<\frac{1}{2}+\varepsilon$. Note that in this construction all maximal anchored squares are disjoint from the top side of $U$. Under this constraint, the upper bound $\frac{1}{2}$ is always attained.

We call a point set $S$ \emph{trivial} if $R(S)$ is disjoint from one of the sides of $U$. The following lemma shows that $\area(R(S))\geq \frac{1}{2}$ for trivial instances.

\begin{lemma}\label{lem:gap-trivial}
If $R(S)$ does not touch one of the sides of $U$, then $\area(R(S))\geq \frac{1}{2}$.
\end{lemma}
\begin{proof}
Without loss of generality, $R(S)$ does not touch the top side of $U$. Let $s=(x,y)$ be a point in $S$ with maximum $y$-coordinate. Consider the maximal empty squares whose lower-left and lower-right corners are at $s$. Since these squares do not touch the top side of $U$, and $s$ has maximum $y$-coordinate, they touch the left and right side of $U$, respectively. Consequently, their combined area is $x^2+(1-x)^2\geq (\frac12)^2+(\frac12)^2=\frac12$. Hence $\area(R(S))\geq \frac{1}{2}$, as claimed.
\end{proof}

\subparagraph{Remark.}
We do not know of any nontrivial point set $S$ for which $\area(R(S))$ attains the lower bound $\frac12$.
Our best lower bound construction for nontrivial instances yields $\frac{4}{7}$; see Fig.~\ref{fig:upperBound}(c).

\subparagraph{Outline.}
In the remainder of Section~\ref{sec:lowerbound}, we consider nontrivial instances $S\subset U$.
A \emph{gap} is a connected component of $U\setminus R(S)$, i.e., of the complement of the reach.
Section~\ref{ssec:basics} presents basic properties of $R(S)$ and its gaps, Section~\ref{ssec:classify}
classifies the possible gaps into five types, and Section~\ref{ssec:charging} presents a charging
scheme in which we define for every gap $C$ a region $R_C\subset R(S)$ such that $\area(C)\leq \area(R_C)$,
and the regions $R_C$ are pairwise interior-disjoint. Summation over all gaps yields
$\area(U\setminus R(S))\leq \sum_{C}\area(R_C)\leq\area(R(S))$, consequently $\area(R(S))\geq \frac12\area(U)=\frac12$.

\subsection{Properties of the Reach and its Gaps}
\label{ssec:basics}

\begin{lemma}	\label{lem:connected}
For every finite set $S\subset U$, the reach $R(S)$ is connected.
\end{lemma}
\begin{proof}
Let $S\subset U$ be a finite set, and let $c_0=(\frac12,\frac12)$ denote the center of $U$.
We show that for each $s\in S$, there is an empty square $Q_s$ anchored at $s$ that contains $c_0$
\textbf{or} whose boundary contains an anchor $s'\in S$ such that $\|s'-c_0\|_\infty <\|s-c_0\|_\infty$ (i.e., $s'$ is closer to $c_0$ in $L_\infty$ norm than $s$). This implies that $Q_s$ (hence $R(S)$) contains a line segment from $s$ to $c_0$ or to $s'$. Consequently, $R(S)$ contains a polyline from every $s\in S$ to $c_0$. By the definition of $R(S)$, this further implies that $R(S)$ contains a polyline between any two points in $R(S)$.

It remains to prove the claim. Let $s\in S$. We may assume without loss of generality that $x(s)\leq y(s)\leq \frac12$, hence $\|s-c_0\|_\infty=\frac12-x(s)$. Let $Q_s$ be the maximal empty square whose lower-left corner is $s$. Refer to Fig.~\ref{fig:gap-proof}(a). If $c_0\in Q_s$, then our proof is complete.
Otherwise, the side length of $Q_s$ is $a_s<\frac12-x(s)$, and there is a point $s'$ in the right or the top side of $Q_s$. The anchor $s'$ lies in the interior of the $L_\infty$-ball of radius $\frac12-x(s)$ centered at $c_0$, hence $\|s'-c_0\|_\infty <\|s-c_0\|_\infty$, as claimed.
\end{proof}

\begin{lemma}	\label{lem:gap-seg}
For every point $p\in U\setminus R(S)$, there exists a point $r\in \partial U$
such that the line segment $pr$ is horizontal or vertical; and $pr\subset U\setminus R(S)$.
\end{lemma}
\begin{proof}
Let $p\in U\setminus R(S)$, and let $Q$ be the maximal empty axis-aligned square centered at $p$. Refer to Fig.~\ref{fig:gap-proof}(b). The boundary of this square, $\partial Q$, intersects $S$ or $\partial U$, otherwise $Q$ would not be maximal.

First assume that $\partial Q$ contains a point $s\in S$. Without loss of generality, we may assume that $x(p)\leq x(s)$ and $y(p)\leq y(s)$. Since $Q$ is empty, the maximal anchored square with upper-right corner at $s$ contains $p$, hence $p\in R(S)$, contradicting our assumption that $p\notin R(S)$.

\begin{figure}[h]
	\centering
	\graphicspath{{./fig/}}
\includegraphics[width=.9\textwidth]{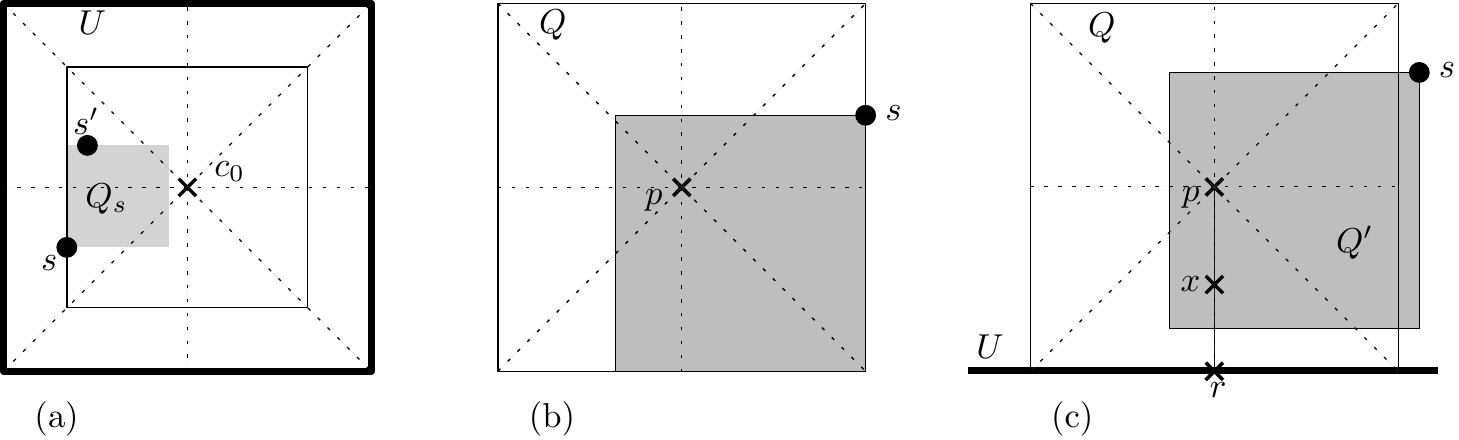}
	\caption{(a) A point $s\in S$ where the anchored square $Q_s$ does not contain $c_0$.
             (b) If $s\in \partial Q$, $x(p)\leq x(s)$, and $y(p)\leq y(s)$,
                then the lower-left square anchored at $s$ contains $p$.
             (c) $\partial Q$ intersects the bottom side of $U$, and $r$ lies below $p$.}
	\label{fig:gap-proof}
\end{figure}

We can now assume that $\partial Q$ intersects $\partial U$. Without loss of generality, $\partial Q\cap \partial U$ lies in the bottom side of both $Q$ and $U$. Let $r\in \partial Q \cap \partial U$ be a point vertically below $p$ (see Fig.~\ref{fig:gap-proof}(c) for an example). Suppose that segment $pr$ intersects $R(S)$. Then some point $x\in pr$ lies in a square $Q'$ anchored at a point $s\in S$. Since $Q$ is empty, the anchor $s$ lies outside of $Q$, and so the side length of $Q'$ is at least half of that of $Q$, i.e., the side length of $Q'$ is at least $|pr|$. However, then $y(s)\geq |pr|$, and the square $Q'$ contains the segment $px$, contradicting our assumption that $p\notin R(S)$. Therefore there is no such point $x\in pr$, and
$pr\subset U\setminus R(S)$, as claimed.
\end{proof}

\begin{corollary}\label{cor:connected1}
The reach is simply connected.
\end{corollary}
\begin{proof}
By Lemma~\ref{lem:connected}, $R(S)$ is connected. Suppose that $R(S)$ is not simply connected. Then there is a gap $C\subset U\setminus R(S)$ such that $\partial C\subset R(S)$. Let $p\in {\rm int}(C)$ be an arbitrary point in the interior of $C$. By Lemma~\ref{lem:gap-seg}, there is a point $r\in \partial U$ such that $pr\subset U\setminus R(S)$, which implies $r\in \partial C$, contradicting our assumption $\partial C\subset R(S)$. Therefore $R(S)$ is simply connected, as required.
\end{proof}

\subsection{Classification of Gaps}
\label{ssec:classify}

In this section we classify the possible shapes of the gaps in $U\setminus R(S)$ for nontrivial instances.
To simplify our analysis, we assume that $S\subset {\rm int}(U)$ and no two points in $S$ have the same $x$- or $y$-coordinates.
This assumption is justified by the following lemma.
\begin{lemma}\label{lem:genpos}
If $\area(R(S))\geq \frac{1}{2}$ for every finite point set $S\subset U$ such that $S\subset {\rm int}(U)$ and no two points in $S$ have the same $x$- or $y$-coordinates, then $\area(R(S))\geq \frac{1}{2}$ for every finite point set $S\subset U$.
\end{lemma}
\begin{proof}
Let $S\subset U$ be a finite point set that contains a point in $\partial U$ or two points with the same $x$- or $y$-coordinate.
Let $\eps_0$ be minimum positive difference between $x$- and $y$-coordinates of points in $S$. For every $\eps\in (0,\eps_0/2)$, translate each point in $S$ by a random vector of length at most $\eps$ into ${\rm int}(U)$. The resulting point set $S_\varepsilon$ lies in ${\rm int}(U)$ and have distinct $x$- and $y$-coordinates with probability 1; the side length of each maximal anchored square may increase by at most $2\varepsilon$, but could decrease substantially. Consequently, ${\rm area}(R(S_\varepsilon))\leq {\rm area}(R(S))+4n\varepsilon$, hence $\lim_{\varepsilon\rightarrow 0}{\rm area}(R(S_\varepsilon))\leq {\rm area}(R(S))$.
\end{proof}

We distinguish a \emph{corner gap}, which is incident to a corner of $U$;
and a \emph{side gap}, which is adjacent to exactly one side of $U$.
We show that every gap is bounded by $\partial U$ and by squares anchored at up to three points in $S$.
We define five types of gaps (two types of corner gaps and three types of side gaps). Each type is defined together with an empty rectangle $B\subset U$ and 1--3 anchors on the boundary of $B$. In each case, the gap is determined by the maximal empty squares that lie entirely in $B$ and are anchored at points in $S\cap B$.

We describe each type modulo the symmetry group of $U$ (i.e., the dihedral group $D_4$). Specifically, we restrict ourselves to corner gaps incident to the lower-left corner of $U$, and side gaps adjacent to the bottom side of $U$. Reflection in the line $x=y$ (resp., $x=\frac12$) maintains corner gaps incident to the origin (resp., side gaps along the bottom side of $U$); and we describe only one variant modulo reflection.

\begin{figure}[h]
	\centering
\includegraphics[width=.95\textwidth]{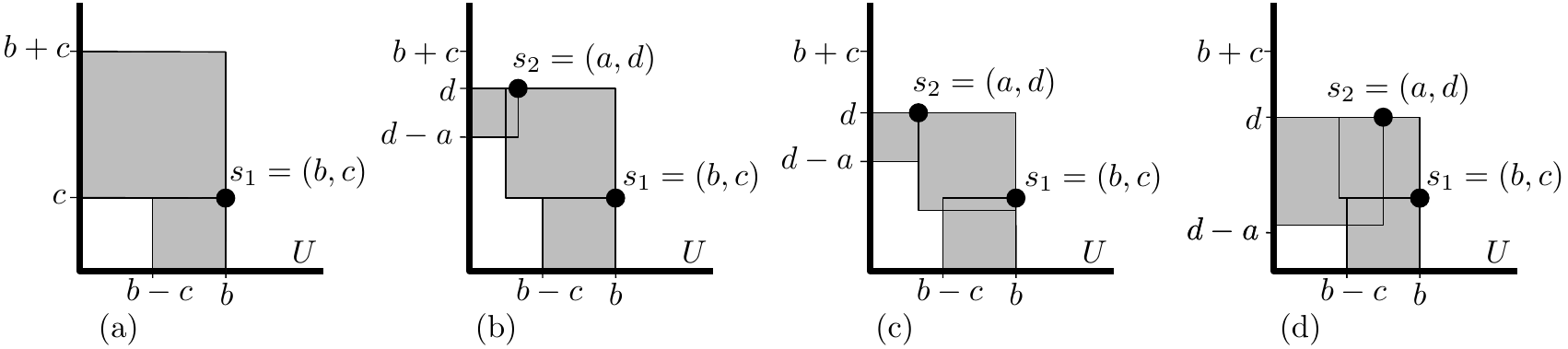}
	\caption{(a) A corner gap of type~1.
             (b--d) Corner gaps of type~2.}
	\label{fig:corner-gaps}
\end{figure}

\begin{enumerate}
\item Let $0<c<b<1$. If $B=[0,b]\times [0,b+c]$ is empty and $s_1=(b,c)\in S$, then
the squares anchored at $s_1$ form a corner gap $[0,b-c]\times [0,c]$. See Fig.~\ref{fig:corner-gaps}(a).
\item Let $0<a<b<1$ and $0<c<d<1$ such that $c<b$ and $d<b+c$. If $B=[0,b]\times [0,d]$ is empty and $s_1=(b,c), s_2=(a,d)\in S$, then the squares anchored at $s_1$ and $s_2$ form a corner gap
    $[0,b-c]\times [0,\min(c,d-a,d-b+a)]\cup [0,\min(a,b-c,b-d+c)]\times [0,d-a]$.
    See Fig.~\ref{fig:corner-gaps}(b--d).
\item Let $0<a<b<1$ and $0<c,d<1$ with $\max(c,d)<b-a$. If $B=[a,b]\times [0,\min(c,d)+(b-a)]$ is empty and $s_1=(a,c), s_2=(b,d)\in S$, then the squares anchored at $s_1$ and $s_2$ form a side gap
    $[a+c,b-d]\times [0,\min(c,d)]$. See Fig.~\ref{fig:side-gaps}(a).
\item Let $0<a<b<1$ and $0<c<d<1$ with $b-a<d$. If $B=[a,b+d]\times [0,d]$ is empty and $s_1=(a,c), s_2=(b,d)\in S$, then the squares anchored at $s_1$ and $s_2$ form a side gap
    $[a+c,b]\times [0,\min(c,d-b+a)]$. See Fig.~\ref{fig:side-gaps}(b).
\item Let $0<a<b<a'<1$ and $0<c<c'<d<1$ with $b-a<d$ and $a'-b<d$. If $B=[a,a']\times [0,d]$ is empty and $s_1=(a,c), s_2=(b,d), s_3=(a',c') \in S$, then the squares anchored at $s_1$, $s_2$, and $s_3$
    form a side gap
    $[a+c,\min(b,a'-d+c')]\times [0,\min(c,d-b+a)]\cup [\min(b,a'-d+c'),a'-c']\times [0,\min(c',d-a'-b)]$.
    See Fig.~\ref{fig:side-gaps}(c) for an example.
\end{enumerate}

\begin{figure}[h]
	\centering
\includegraphics[width=.95\textwidth]{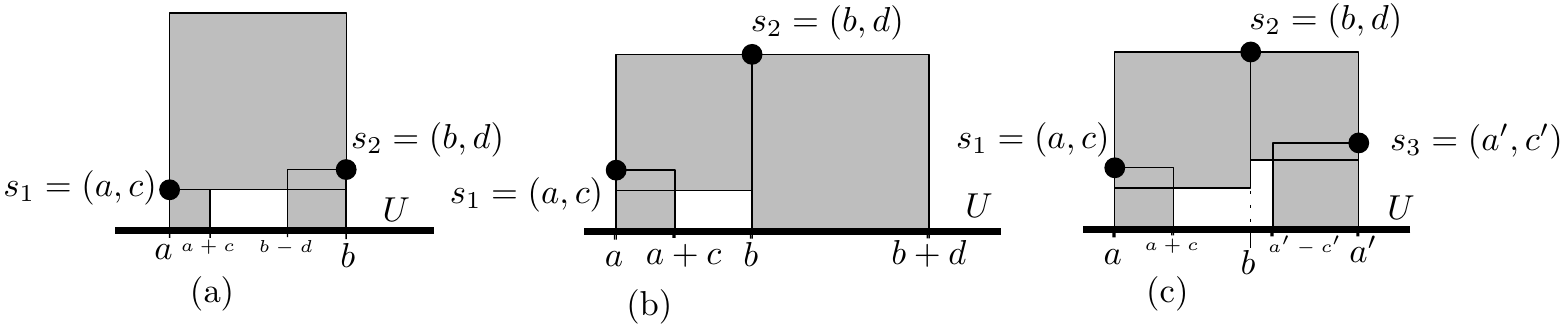}
	\caption{(a--c) Side gaps of type~3, 4, and 5, respectively.}
	\label{fig:side-gaps}
\end{figure}

\begin{lemma}\label{lem:classify-boundary}
Every gap $C$ of type~1--5 is disjoint from all empty squares that are anchored at points in $S$ and lie in the exterior of the defining box $B$ of $C$. Consequently, $C$ is bounded by $\partial U$ and some empty squares anchored at points in $S\cap B$.
\end{lemma}
\begin{proof}
In each of the five cases, $\partial B\cap {\rm int}(U)$ is covered by empty squares anchored at the points in $S$ that define $C$. More precisely, each point in $\partial B\cap {\rm int}(U)$ lies in an empty square anchored at a point in $S\cap \partial B$ blocked by some point in $S\cap \partial B$ or $\partial U\cap \partial B$.
For $s_i\in S$ lying in the exterior of $B$, let $Q_i$ be a square anchored at $s_i$. If $Q_i$ intersects $B$, then its interior intersects $\partial B\cap {\rm int}(U)$, hence it intersects a square $Q_j$ anchored at some $s_j\in \partial B$ and blocked by some point $p_j \in S\cap \partial B$ or $\partial U\cap \partial B$. Since ${\rm int}(Q_i)$ contains neither $s_j$ nor $p_j$, we have $Q_i\cap B\subset Q_j$, and so $Q_i$ is disjoint from the gap $C$, as claimed.
\end{proof}

We prove the following classification result.
\begin{lemma}\label{lem:classify}
Every gap of a nontrivial instance is of one of the five types defined above.
\end{lemma}

\begin{proof}
	Let $C$ be a gap. Since both $C$ and $R(S)$ are connected (cf.\ Lemma~\ref{lem:connected}), $C\cap \partial U$ is also connected. Since $S$ is nontrivial, $C$ touches at most two sides of $\partial U$.
	We distinguish two cases, depending on whether $C$ is incident to a corner of $U$.
	
	\subparagraph{Corner gaps.}
	Assume that $C$ is incident to the lower-left corner of $U$ (i.e., the origin $\mathbf{0}$). We need to show that $C$ is of type~1 or 2. Let $[0,b]^2$ be the maximal empty square incident to $\mathbf{0}$. We may assume (by a reflection in $y=x$, if necessary) that there is an anchor $s_1=(b,c)$ on the right side of this square, for some $c\in (0,b)$. If $[0,b]\times [0,b+c]$ is empty, then $C$ is of type~1. Otherwise, let $s_2=(a,d)$ be the anchor with minimum $y$-coordinate in the interior of $[0,b]\times [0,b+c]$. Then $[0,b]\times [0,d]$ is empty, and $C$ is of type~2.
	
	\subparagraph{Side gaps.}
	Assume that $C$ is not incident to any corner of $U$, but it touches the bottom side of $U$.
	By Corollary~\ref{cor:connected1}, $C\cap \partial U$ is a horizontal line segment, that we denote by $p_1p_2$, where $x(p_1)<x(p_2)$. By the definition of $R(S)$, point $p_1$ is the lower-right corner of a square $q_1$ anchored at some point $s_1\in S$; and $p_2$ is the lower-left corner of a square $q_2$ anchored at some $s_2\in S$. Clearly, $s_1$ (resp., $s_2$) is a upper-left or upper right corner of $q_1$ (resp., $q_2$).
	
	We claim that $x(p_1)\neq x(s_1)$ or $x(p_2)\neq x(s_2)$. Suppose, to the contrary, that $x(p_1)= x(s_1)$ and $x(p_2)= x(s_2)$. Without loss of generality, assume that $y(s_1)\leq y(s_2)$. If the rectangle $R=[x(p_1),x(p_2)]\times [0,y(s_1)]$ is empty, then the lower-right square anchored at $s_1$ contains segment $p_1p_2$. Otherwise, let $s_3$ be the point in the interior of $R$ with minimum $y$-coordinate. Since $x(p_1)<x(s_3)<x(p_2)$, the lower-left and lower-right squares anchored at $s_3$ reach $p_1p_2$. In both cases, we see that some anchored squares touch $p_1p_2$, contradicting the assumption that this segment is in a gap. This completes the proof of the claim.
	
	We may assume (by a reflection in $x=\frac12$) that $x(p)\neq x(s_1)$. Then $s_1$ is the upper-left and $p_1$ is the lower-right corner of the anchored square $q_1$. Letting $s_1=(a,c)$, we have  $p_1=(a+c,0)$. Let $Q_0=[a,a+h]\times [0,h]$ be the maximum empty square with lower-left corner at $(a,0)$.
	Note that this square cannot touch the right side of $U$, otherwise it would contain $s_2$. Therefore,
	there is an anchor $s=(b,d)$ in the right or top side of $Q_0$. We distinguish two cases:
	
	\subparagraph{Case~1: $s$ is in the right side of $Q_0$.}
	Since the lower-left square anchored at $s$ does not contain $p_1=(a+c,0)$, we have $a+c<x(s)-y(s)$, hence $a+y(s)<x(s)-c$. Expand $Q_0$ upward into a maximal empty rectangle $R_1=[a,a+h]\times [0,k]$ for some $k\in (h,1]$. If $R_1$ contains $[a,b]\times [0,\min(a,b)+(b-a)]$, then $C$ is of type~3.
	Otherwise, it is of type~5.
	
	\subparagraph{Case~2: $s$ is in the top side of $Q_0$.}
	Expand $Q_0$ to the right into a maximal empty rectangle $R_2=[a,a+k]\times [0,h]$ the for some $k\in (h,1-a)$.
	If $R_2$ contains $[a,b+d]\times [0,d]$, then $C$ is of type~4. Otherwise, it is of type~5.
\end{proof}

It is now easy to check that the following properties hold for all five types of gaps.

\begin{corollary}\label{cor:classify-squares}
\begin{itemize}
\item[]
\item[{\rm (i)}] Each gap is either a rectangle incident to a side of $U$, or the union of two rectangles incident to the same side of $U$ (which we call an \emph{L-shaped} gap).
\item[{\rm (ii)}] Every edge $uv$ of a gap is contained in either $\partial U$ or a maximal anchored rectangle.
 Consequently, the square built on the side $uv$ outside of the gap lies either outside of $U$ or in $R(S)$.
\end{itemize}
\end{corollary}

We say that a point $p\in \partial U$ is a \emph{lead} if it is a vertex of a maximal anchored square $q_i^j$, and $ps_i$ is a diagonal of $q_i^j$. We observe that one or two vertices of a gap along $\partial U$ is a lead.

\begin{corollary}\label{cor:lead}
\begin{itemize}
\item[]
\item If $C$ is a rectangular gap, then at least one endpoint of $C\cap \partial U$ is a lead,
\item otherwise both endpoints of $C\cap \partial U$ are leads.
\end{itemize}
\end{corollary}

\subsection{Charging Scheme}
\label{ssec:charging}

For every gap $C$, we define a region $R_C \subset R(S)$; and then we show that
$\area(R_C) \geq \area(C)$ and the regions $R_C$ are pairwise interior-disjoint.

For ease of exposition, we subdivide every L-shaped side gap $C$ into two rectangles $C=C_1\cup C_2$,
then define interior-disjoint regions $R_{C_1}$ and $R_{C_2}$, and let $R_C:=R_{C_1}\cup R_{C_2}$.
Specifically, let $\mathcal{C}^*$ be a set of regions that contains:
(1) all corner gaps,
(2) all rectangular side gaps, and
(3) for each L-shaped side gap $C$, the two interior-disjoint rectangles $C_1$ and $C_2$,
such that $C=C_1\cup C_2$ and both $C_1$ and $C_2$ have a common side with $\partial U$
(see Fig.~\ref{fig:gapdiv} for an example).
By Corollary~\ref{cor:lead}, at least one vertex of every rectangle in $\mathcal{C}^*$ is a lead,
and two vertices of every L-shaped corner gap in $\mathcal{C}^*$ are leads.
%

\begin{figure}[h]
	\centering
\includegraphics[width=.8\textwidth]{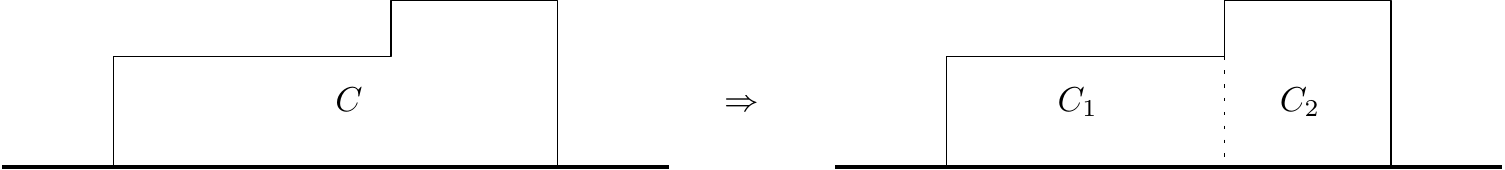}
	\caption{An L-shaped side gap $C$ is subdivided into two rectangles $C_1,C_2\in \mathcal{C}^*$.}
	\label{fig:gapdiv}
\end{figure}

We are now ready to define a region $R_C$ for each region $C\in \mathcal{C}^*$.
\begin{itemize}
\item[]
\item Let $C=(a,b,c,d)$ be a rectangle in $\mathcal{C}^*$. Assume w.l.o.g. that $bc$ is contained in the bottom side of $U$, and $c$ is a lead (a symmetric construction applies if $bc$ is contained in another side of $U$ or $b$ is the only lead). Refer to Fig.~\ref{fig:charged-region}(a).
    Let $\ell_1$ and $\ell_2$ be lines of slope 1 passing through $a$ and $c$, respectively.
    Let $p_1$ be the intersection of $\ell_1$ with the vertical line through $cd$, and let $p_2$ be the intersection of $\ell_2$ with the horizontal line through $da$. Let $z_1$ (resp., $z_2)$ be the intersection point of $\ell_2$ (resp., $\ell_1$) with the line of slope $-1$ passing through $p_1$ (resp., $p_2)$. Then $R_C$ is the smaller pentagon out of  $(a,d,c,z_1,p_1)$ and $(a,d,c,p_2,z_2)$.
\item Let $C=(a,b,c,d,e,f)$ be a L-shaped corner gap in $\mathcal{C}^*$. Assume w.l.o.g. that $b$ is the lower-left corner of $U$. By Corollary~\ref{cor:lead}, both $a$ and $c$ are leads. Refer to Fig.~\ref{fig:charged-region}(b).
    Let $\ell_1$ and $\ell_2$ be lines of slope 1 passing through $a$ and $c$, respectively.
    Let $p_1,p_2\in S$ be the anchors on $\ell_1$ and $\ell_2$, respectively (which exist since both $a$ and $c$ are leads). Let $z_1$ (resp., $z_2)$ be the intersection point of line $\ell_2$ (resp., $\ell_1$) with the line of slope $-1$ passing through $p_1$ (resp., $p_2)$. Then $R_C$ is the smaller heptagon out of  $(a,f,e,d,c,z_1,p_1)$ and $(a,f,e,d,c,p_2,z_2)$.
\end{itemize}

\vspace{-\baselineskip}

\begin{figure}[htbp]
	\centering
\includegraphics[width=.8\textwidth]{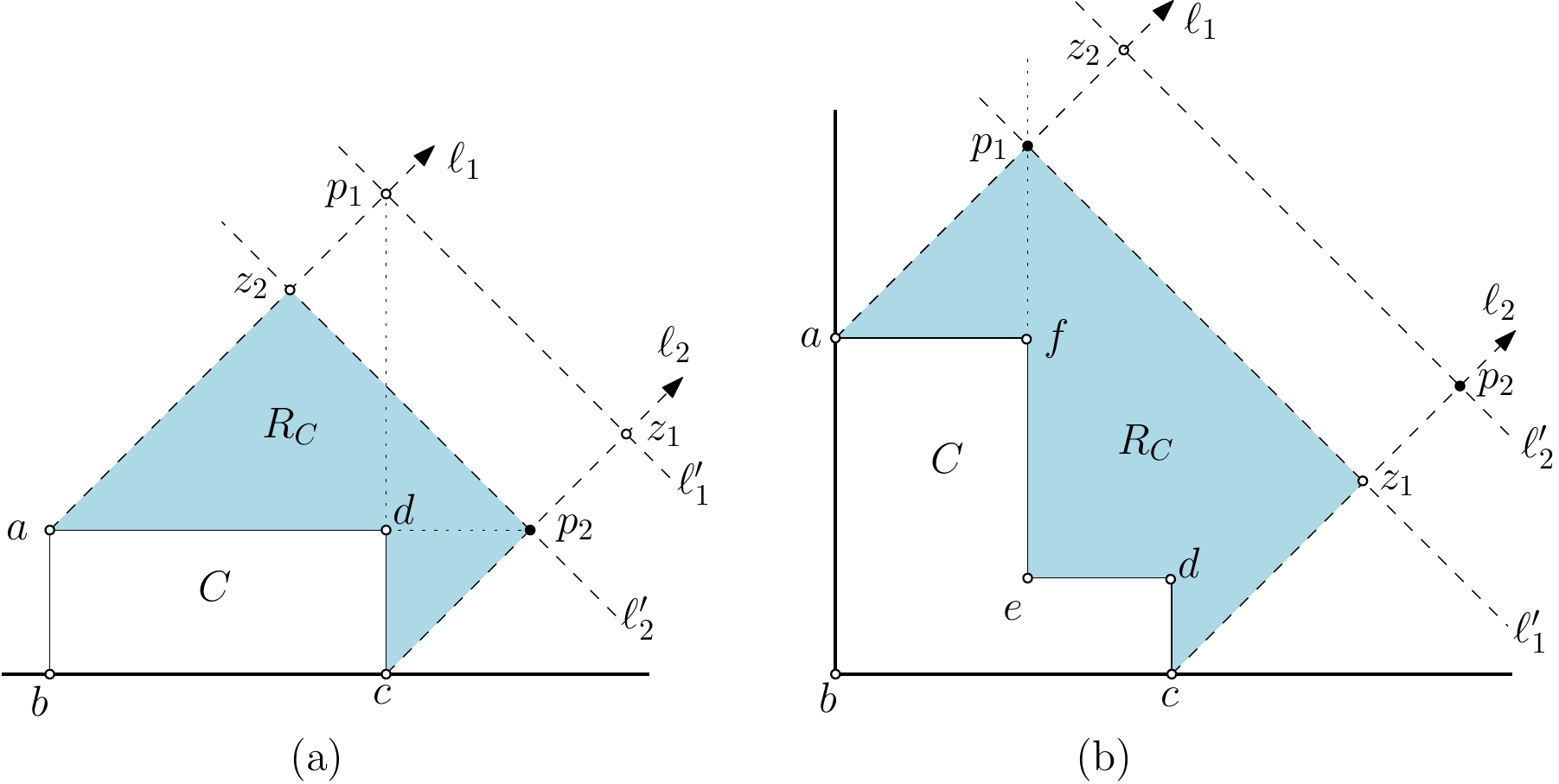}
	\caption{Region $R_C$.
    (a) $C\in \mathcal{C}^*$ is a rectangle.
    (b) $C\in \mathcal{C}^*$ is an L-shaped corner gap.}
	\label{fig:charged-region}
\end{figure}

\begin{lemma}\label{lem:properties}
For every $C\in \mathcal{C}^*$, we have:
\begin{enumerate}[label = (P\arabic*)]
\item\label{p:1} $R_C\subseteq R(S)$,
\item\label{p:2} ${\rm int}(R_C)$ does not contain any anchors, and
\item\label{p:3} $\area(C)\leq \area(R_C)$.
\end{enumerate}
\end{lemma}
\begin{proof}
The region $C\in \mathcal{C}^*$ is either a gap or a rectangle within an L-shaped side gap; Fig.~\ref{fig:charged-region}(b). Let $C^*$ be the gap that contains $C$,
and $B$ the box defining the gap $C^*$. For all five types of gaps, ${\rm int}(B)$ does not contain any anchor, and $B\setminus C^*\subset R(S)$. By Corollaries~\ref{cor:classify-squares}(ii) and \ref{cor:lead}, the points $p_1$ and $p_2$ lie in $B$. Consequently, $R_C\subset B$, hence $R_C\subset B\setminus C^*$. This confirms (P1) and (P2).

To prove (P3), we distinguish two cases. First assume that $C$ is an $x\times y$ rectangle.
Let $T$ be an isosceles right triangle whose hypotenuse has length $x+y$.
It is easy to check that $\area(C) \leq \area(T)$. Indeed,
$\area(T)
 = \left(\frac{1}{2}(x+y)\right)^2
 = \frac{1}{2}\left(\frac{x^2}{2} + xy + \frac{y^2}{2}\right)
 \geq xy
 =\area(C)$.
By definition, $R_C$ contains a triangle congruent to $T$, consequently
$\area(C) \leq \area(T)\leq \area(R_C)$, as claimed.

Next assume that $C\in \mathcal{C}^*$ is an L-shaped corner gap; Fig.~\ref{fig:charged-region}(b).
Assume that $C$ is formed by three interior-disjoint axis-aligned rectangles
defined by diagonals $ae$, $be$, and $ce$.
Let their dimensions respectively be $x\times y$, $x\times z$, and $w\times z$.
Let $T_1$ and  $T_2$ be isosceles right triangles whose hypotenuses are of length $x+y$ and $w+z$, respectively.
Let $T_3$ and $T_4$ be isosceles right triangles whose legs are of length $x$ and $z$, respectively.
By definition, $R_C$ contains interior-disjoint triangles congruent to $T_1$, $T_2$, $T_3$, and $T_4$: the hypotenuses of the respective triangles are in the same supporting lines as $ef$, $ed$, $ap_1$, and $cp_2$ respectively.
Using the same argument as in the previous case, we can show that $\area(T_1)$ and $\area(T_2)$ are, respectively,
greater or equal than the areas of the $x\times y$ and $w\times z$ rectangles.
It remains to show that $\area(T_3)+\area(T_4)$ is greater or equal than the area of the $x\times z$ rectangle.
By definition, we have $\area(T_3)+\area(T_4)=\left(\frac{x^2}{2} + \frac{z^2}{2}\right)\ge xz$ for all $x,z>0$.
\end{proof}



For every $C\in \mathcal{C}^*$, the region $C\cup R_C$ is a convex pentagon, that we denote by $P_C$.
The definition of $R_C$ immediately implies the following properties:
\begin{observation}\label{obs:pent1}
	\begin{enumerate}
		\item[]
		\item Two consecutive sides of $P_C$ are common with of $C$, hence they are axis-parallel.
		\item Three consecutive sides of $P_C$ have slope $\pm 1$, and lie in $R(S)$;
		they are each diagonals of axis-aligned squares that lie in $R(S)$.
		\item $P_C$ and $C$ have three common vertices, one of them is a lead.
	\end{enumerate}
\end{observation}

\begin{figure}[htbp]
	\centering
	\includegraphics[width=.8\textwidth]{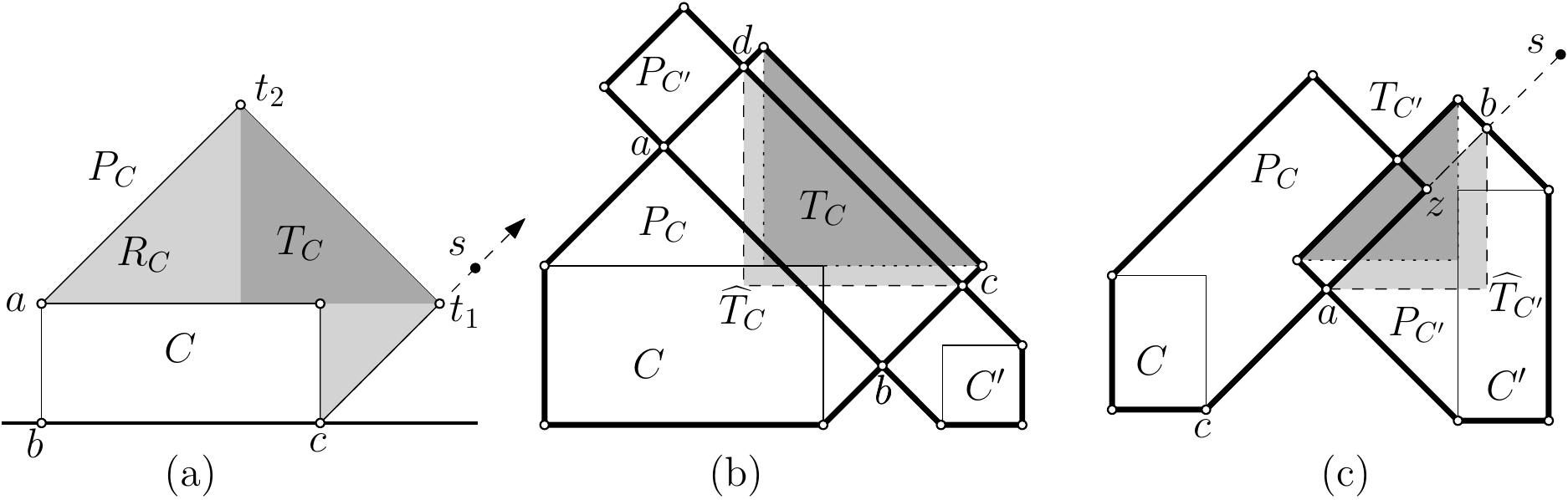}
	\caption{(a) Pentagon $P_C=C\cup R_C$, tips $t_1$ and $t_2$, sentinel edge $t_1t_2$,
		and the sentinel triangle $T_C$. If $c$ is a lead vertex, then $t_1$ is a safe tip.
		(b) If ${\rm int}(R_C)\cap{\rm int}(R_{C'})\neq \emptyset$, then $\widehat{T}_C$ covers part of $C$.
		(c) If $z$ is a safe tip of $P_C$ and $z\in {\rm int}(P_{C'})$, then $\widehat{T}_{C'}$ covers part of $C'$.}
	\label{fig:sentinel}
\end{figure}

We introduce some terminology. The two vertices of $P_C$ that are not incident to $C$ are called the \emph{tips} of $P_C$. The line segment between the two tips is called the \emph{sentinel side} of $P_C$, and the isosceles right triangle in $P_C$ whose hypotenuse is the sentinel side is called the \emph{sentinel triangle}, denoted $T_C$. See Fig.~\ref{fig:sentinel}(a)

\begin{lemma}\label{lem:sentinel}
	For every $C\in \mathcal{C}^*$,
	the triangle $T_C$ is adjacent to $C$ (i.e., $\partial T_C\cap \partial C\neq \emptyset$).
\end{lemma}
\begin{proof}
	If $C$ is a rectangle, then the claim follows immediately from the definition of $R_C$.
	Assume that $C$ is an L-shaped corner gap of type~2. Then by the definition of $R_C$,
	one of the tips of $P_C$ is an anchor $s\in S$ on the boundary of the box $B$ that defines $C$.
	The sentinel edge lies along the diagonal of a maximal empty square $q$ anchored at $s$.
	The square $q$ is adjacent to $C$ by the definition of gaps of type~2.
\end{proof}

Furthermore, a tip of $P_C$ is called \emph{safe} if it an adjacent vertex of $P_C$ is a lead vertex of $C$.
By Corollary~\ref{cor:lead}, at least one of the two tips $P_C$ is safe; and if $C$ is an L-shaped corner gap,
then both tips are safe.

\begin{lemma}\label{lem:vertex}
	If ${\rm int}(R_C)\cap{\rm int}(R_{C'})\neq \emptyset$ for some $C,C'\in \mathcal{C}^*$, $C\neq C'$,
	then
	\begin{enumerate}
		\item all intersection points in $\partial P_C\cap \partial P_{C'}$ involve edges of slopes $\pm 1$;
		\item a tip of $P_C$ or $P_{C'}$ lies in ${\rm int}(R_{C'})$ or ${\rm int}(R_C)$, respectively.
	\end{enumerate}
\end{lemma}
\begin{proof}
	Note that $P_C$ and $P_{C'}$ each has two axis-parallel sides, which are edges of $C$ and $C'$, respectively. These edges do not cross each other (since the regions in $\mathcal{C}^*$ are interior-disjoint), and they do not cross any edge of $P_C$ and $P_{C'}$ of slope $\pm 1$ (since those edges are contained in $R(S)$). Consequently, all intersection points in $\partial P_C\cap \partial P_{C'}$ involve edges of slopes $\pm 1$.
	This proves the first claim.
	
	Note also that the regions in $\mathcal{C}^*$ are pairwise interior-disjoint and they all lie in $U\setminus R(S)$. The pentagon $P_C$ (resp., $P_{C'}$) has three common vertices with $C$ (resp., $C'$), and these vertices are certainly not in ${\rm int}(P_{C'})$ (resp., ${\rm int}(P_{C})$). Therefore, it is enough to show that a vertex of pentagon $P_C$ or $P_{C'}$ lies in the interior of the other pentagon.
	
	Suppose, to the contrary, that neither ${\rm int}(P_C)$ nor ${\rm int}(P_{C'})$ contains any vertex of the other pentagon. Let $D=P_C\cap P_{C'}$. Since both pentagons are convex, $D$ is a convex polygon, which is not incident to any vertex of $P_C$ or $P_{C'}$. Therefore, $D$ is a convex quadrilateral in which all vertices are intersection points between some sides of $P_C$ and $P_{C'}$. By claim~1, $D$ is a rectangle
	whose edges are of slope $\pm 1$.
	
	Let $D=(a,b,c,d)$ such that $ab$ and $cd$ lie in parallel edges of $P_{C'}$; Fig.~\ref{fig:sentinel}(b).
	Then $ab$ and $cd$ are translates of the sentinel edge of $P_C$. We may assume w.l.o.g. that $ab$ is not equal to the sentinel edge of $P_C$, and so $ab$ separates the sentinel edge of $P_C$ from $C$.
	By Lemma~\ref{lem:sentinel}, the sentinel triangle $T_C$ is adjacent to $C$. Translate $T_C$ to
	a right triangle $\widehat{T}_C$ of hypotenuse $ab$. Then ${\rm int}(C)\cap {\rm int}(\widehat{T}_C)\neq \emptyset$. However, by Observation~\ref{obs:pent1}(2), $\widehat{T}_C\subset R(S)$, contradicting
	our assumption that $C$ is part of a gap.
\end{proof}

We are now ready to show that the regions $R_C$, $C\in \mathcal{C}^*$, are interior-disjoint.

\begin{lemma}\label{lem:disjoint}
	For every two regions $C,C'\in \mathcal{C}^*$, $C\neq C'$,
	we have ${\rm int}(R_C)\cap {\rm int}(R_{C'})=\emptyset$.
\end{lemma}
\begin{proof}
	Suppose, to the contrary, that there are two regions $C,C'\in \mathcal{C}^*$, $C\neq C'$,
	such that ${\rm int}(R_C)\cap {\rm int}(R_{C'})\neq\emptyset$. By Lemma~\ref{lem:vertex}(2),
	a tip of $P_C$ or $P_{C'}$ lies in the interior of $R_{C'}$ or $R_C$.
	We distinguish two cases.
	
	\subparagraph{Case~1. A safe tip of $P_C$ lies in ${\rm int}(R_{C'})$,
		or a safe tip of $P_{C'}$ lies in ${\rm int}(R_{C})$.}
	Assume w.l.o.g. that $z$ is a safe tip of $P_C$ and $z\in {\rm int}(R_{C'})$.
	Assume further that $c$ is a lead vertex of $C$ lying on the bottom side of $U$, and $cz$ is a side of of $P_C$ of slope $1$. Since $c$ is a lead, there exists an anchor $s\in S$ such that $cs$ is the diagonal
	of a maximal square anchored at $s$. Consequently, $cz\subseteq cs$, and $cs\subset R(S)$.
	
	Note that $c\not\in {\rm int}(P_{C'})$ since $c\in \partial U$, and $s\not\in {\rm int}(P_{C'})$ by \ref{p:2},
	but $z\in {\rm int}(R_{C'})$. Therefore segment $cs$ crosses the boundary of $P_{C'}$ twice.
	Since $cs\subset R(S)$ and has slope $1$, it cross two parallel edges of $P_{C'}$ of slope $-1$.
	Let $ab=cs\cap P_{C'}$, and note that $ab$ is a translate of the sentinel edge of $P_{C'}$.
	By Lemma~\ref{lem:sentinel}, the sentinel triangle $T_{C'}$ is adjacent to $C'$. Translate $T_{C'}$ to
	a right triangle $\widehat{T}_{C'}$ of hypotenuse $ab$. Then ${\rm int}(C')\cap {\rm int}(\widehat{T}_{C'})\neq \emptyset$. However, by Observation~\ref{obs:pent1}(2), $\widehat{T}_{C'}\subset R(S)$, contradicting
	our assumption that $C'$ is part of a gap.
	
	\subparagraph{Case~2. None of safe tips of $P_C$ and $P_{C'}$, resp., lies in ${\rm int}(R_{C'})$ and ${\rm int}(R_{C})$.} We may assume w.l.o.g. that $z$ is an unsafe tip of $P_C$ that lies in ${\rm int}(R_{C'})$.
	By Corollary~\ref{cor:lead}, $P_C$ has at least one safe tip, and if it has only one safe tip,
	then $C$ is a rectangle. Therefore, $z$ is adjacent to a safe tip in $P_C$,
	and by the definition of safe tips, the second neighbor of $z$ is a lead vertex, which lies in $\partial U$.
	It follows that the side of $P_C$ opposite to $z$ is contained in $\partial U$.
	We may assume that $P_C=(a,b,c,t,z)$, $C=(a,b,c,d)$, where $c$ is a lead vertex, $t$ is a safe tip, and $bc$ is contained in the bottom side of $U$; refer to Fig,~\ref{fig:opposite}(b)-(c).
	
	We claim that $P_{C'}$ also has an unsafe tip $z'$ that lies in ${\rm int}(R_C)$.
	Suppose, to the contrary, and ${\rm int}(P_C)$ does not contain any unsafe tip of $P_{C'}$.
	Then by assumption ${\rm int}(P_{C'})$ contains exactly one vertex of $P_C$, namely $z$,
	and ${\rm int}(P_C)$ does not contain any vertex of $P_{C'}$.
	The two intersection points
	in $\partial P_C\cap \partial P_{C'}$ are between a single side $uv$ of $P_{C'}$ and
	the two sides of $P_C$ incident to $z$. The sides incident to $z$ have slopes $\pm 1$,
	so $uv$ cannot have slope $\pm 1$, contradicting Lemma~\ref{lem:vertex}(1).
	This confirms the claim.
	
	It follows that $C'$ is also a rectangle, and the side of $P_{C'}$ opposite to $z'$ is contained in $\partial U$.
	Let $P_{C'}=(a',b',c',t',z')$, where $t'$ is a safe tip, $c'$ is a lead, and $b'c'\subset \partial U$.
	Since the path $(t,z,a)$ and $(t',z',a')$ cross twice, $b'c'$ is contained in the top side of $U$;
	as indicated in Fig.~\ref{fig:opposite}(b)-(c).
	
	\begin{figure}[htbp]
		\centering
		\includegraphics[width=.9\textwidth]{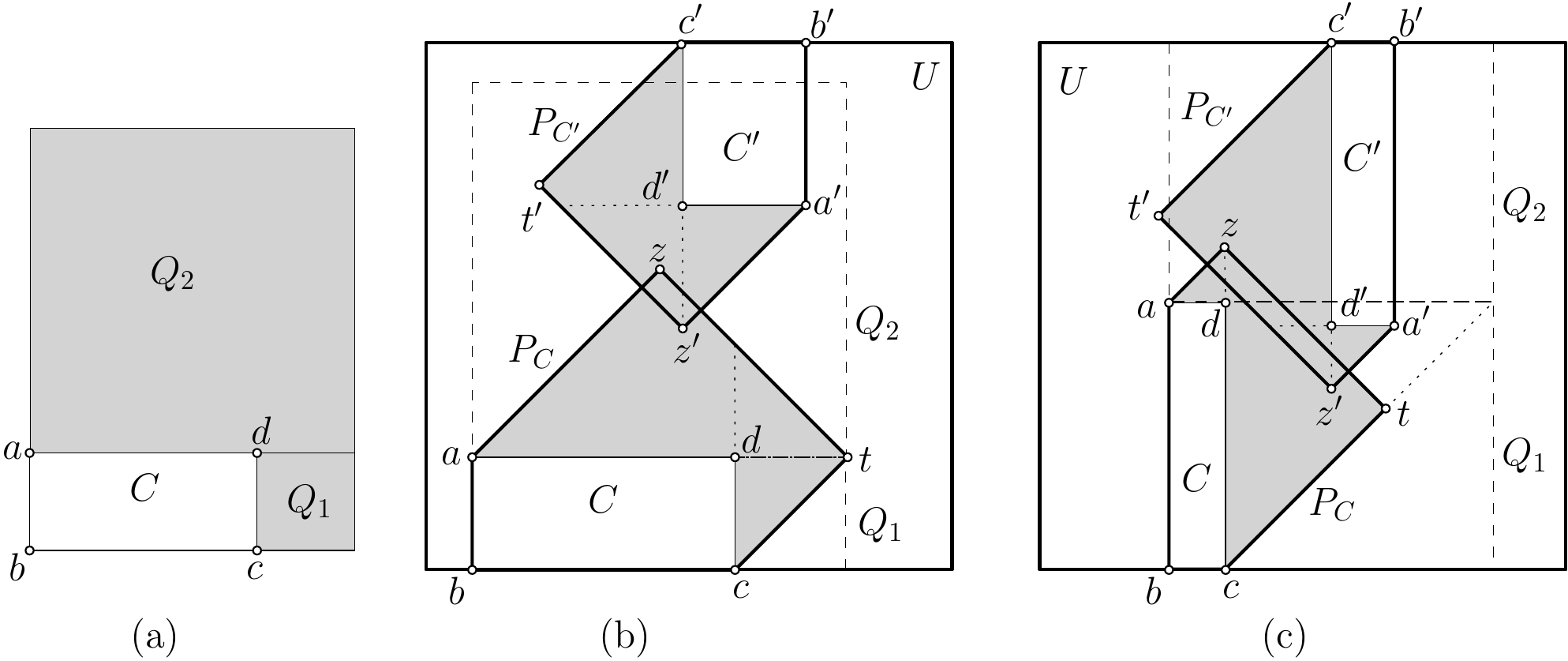}
		\caption{(a) A rectangle $C=(a,b,c,d)$, square $Q_1$ built on $cd$ to the right of $C$,
			and square $Q_2$ built on the top side of $C\cup Q_1$.
			(b)-(c) An unsafe tip $z$ of $P_C$ lies in ${\rm int}(P_{C'})$, and
			an unsafe tip $z'$ of $P_{C'}$ lies in ${\rm int}(P_C)$.}
		\label{fig:opposite}
	\end{figure}

	Since the roles of $C$ and $C'$ are symmetric, we may assume w.l.o.g. that $|a'z'|\leq |az|$.
	Let $Q_1$ be the square built on $cd$ to the right of $C$, and $Q_2$ the square
	built on the top side of the rectangle $C\cup Q_1$; see Fig.~\ref{fig:opposite}(a).
	Both $Q_1 $and $Q_2$ are covered, that is, $Q_1,Q_2\subset R(S)$.
	Indeed, $Q_1\subset R(S)$ by Corollary~\ref{cor:classify-squares}(ii).
	Let $Q$ be the anchored square whose boundary contains top side $ad$ of $C$.
	By the characterization of the gaps, $Q_2 \subseteq Q\subseteq R(S)$.
	
	We next claim that ${\rm int}(Q_2)\cap {\rm int}(C')\neq \emptyset$. By the definition of $R_C$, we have $t\in Q_1$ (possibly, $ct$ is the diagonal of $Q_1$). Consequently, the projection of $C\cup Q_1\cup Q_2$ to the $x$-axis contains that of $P_C$. Let $\ell'$ be a vertical line through $z'$. By the definition of $R_{C'}$, $\ell'\cap C' \neq \emptyset$, hence $\ell'\cap a'd'\neq \emptyset$. Since $z'\in {\rm int}(P_C)$, and $Q_2$ is built on top of $C\sup Q_1$, the line $\ell'$ passes through ${\rm int}(Q_2)$. We show that $\ell'\cap a'd'\in \rm{int}(Q_2)$ by comparing $y$-coordinates: $z'\in {\rm int}(P_C)$ gives $y(z')<y(z)$, and $|a'z'|\leq |az|$ implies $y(a')<y(z)+(y(z)-y(a))=2y(z)-y(a)$. Hence $a'$ lies below the top side of $Q_2$, which yields $\ell'\cap a'd'\in \rm{int}(Q_2)$. Finally, note that a neighborhood of $\ell'\cap a'd'$ intersects ${\rm int}(C')$ and lies in ${\rm int}(Q_2)$. This establishes the claim.
	
	We have proved that $Q_2\subset R(S)$ and ${\rm int}(Q_2)\cap {\rm int}(C')\neq \emptyset$.
	This contradicts the assumption that $C'$ is part of a gap.
	Both cases lead to a contradiction, therefore ${\rm int}(R_C)\cap {\rm int}(R_{C'})=\emptyset$.
\end{proof}


\begin{theorem}\label{thm:area}
    For every finite set $S \subset U$, we have $\area(R(S))\geq \frac{1}{2}$,
    and this bound is the best possible.
\end{theorem}

\begin{proof}
By Lemma~\ref{lem:genpos}, it suffices to prove the lower bound when $S\subset {\rm int}(U)$ and no two points in $S$ have the same $x$- or $y$-coordinate. For all gaps $C\subset U\setminus R(S))$, we have defined interior-disjoint regions $R_C \subset R(S)$ such that $\area(C)\leq \area(R_C)$. Consequently, $\sum_C \area(C)\leq \sum_C\area(R_C)\leq \area(R(S))$, which immediately yields $\area(R(S))=1-\sum_C\area(C)\geq 1-\area(R(S))$, and $\area(R(S))\geq 1/2$, as claimed. This bound is the best possible: the point set $S = \{(\frac12,0)\}$ attains $\area(R(S)) = \frac12$.
\end{proof}

\section{Algorithm for Computing the Reach}
\label{sec:algorithm}

In this section, we show how to compute efficiently the reach of a given point set.

\begin{theorem}\label{thm:algorithm}
For a set $S\subset U$ of $n$ points, $R(S)$ can be computed in $O(n\log n)$ time.
\end{theorem}

Recall that, for a set $S=\{s_1,\ldots , s_n\}$, the reach is defined as a union of $4n$ squares, $R(S)=\bigcup_{i=1}^n \bigcup_{j=1}^4 q_i^j$, where $q_i^1$ is the maximal axis-aligned empty square in $U$ whose lower-left corner is $s_i$, and $q_i^2$, $q_i^3$, and $q_i^4$ are defined similarly where $s_i$ is the upper-left, upper-right, and lower-right corner, respectively. Since any two squares cross in at most two points, the $4n$ squares $q_i^j$ ($i=1,\ldots, n$ and $j=1,\ldots , 4$) form a pseudo-circle arrangement. It is well known that the union of $O(n)$ pseudo-circles has $O(n)$ vertices~\cite{KLP86}. The union of $4n$ axis-aligned \emph{squares} can be computed by a sweep-line algorithm in $O(n\log n)$ time~\cite{Ben77}.

We note that Bentley's sweep-line algorithm can compute the \emph{area} of the union of $n$ axis-aligned \emph{rectangles} in $O(n\log n)$ time (without computing the union itself, which may have $\Theta(n^2)$ complexity). Computing the volume of the union of axis-aligned hyper-rectangles in $\mathbb{R}^d$ is known as \emph{Klee's measure problem}, and the current best algorithms~\cite{Chan13} for $d\geq 3$ run in $O(n^{d/2})$ time in general, and in $O(n^{(d+1)/3}\, {\rm polylog}(n))$ time for hypercubes (see also~\cite{YS15}).

\begin{figure}[h]
	\centering
\includegraphics[width=.8\textwidth]{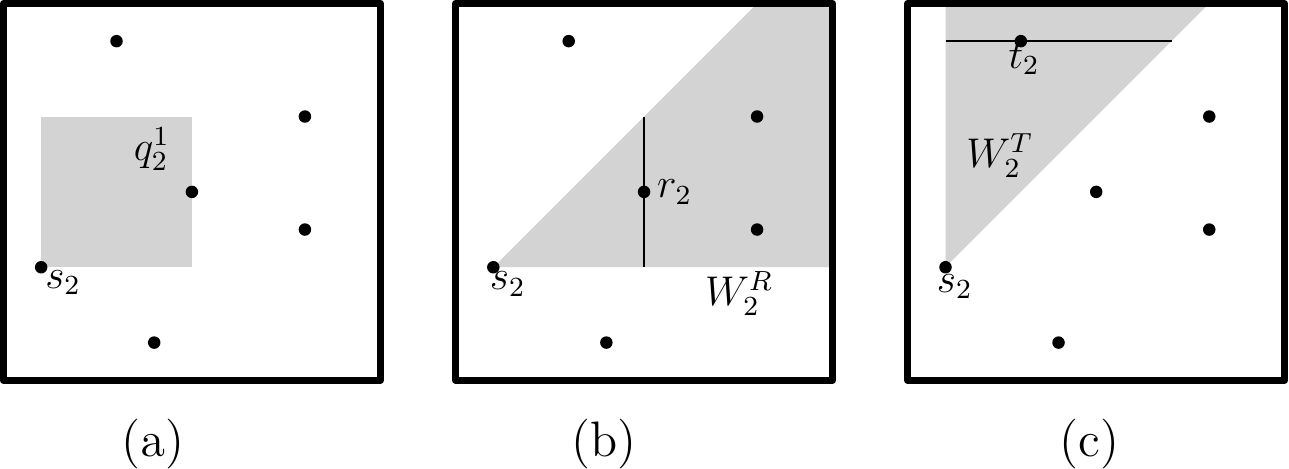}
	\caption{A set $S$ of 6 anchors in $U=[0,1]^2$.
            (a) The maximal empty anchored square $q_2^1$;
            (b) wedge $W_2^R$; and
            (c) wedge $W_2^T$.}
	\label{fig:wedges}
\end{figure}

It remains to compute the $4n$ anchored maximal empty squares $q_i^j$. We focus on the $n$ \emph{lower-left} anchored squares  $q_i^1$ ($i=1,\ldots, n$), the other three types can be computed analogously. For every $i=1,\ldots, n$, the lower-left corner of $q_i^1$ is $s_i$, and its left \emph{or} top side contains another anchor or a point in $\partial U$; we say that this point is the \emph{blocker} of $q_i^1$.
For each $i$, we find a first point that may block the square $q_i^1$ on the left and on the top side, independently. The blocker of $q_i^1$ is the points closest to $s_i$ in $L_\infty$ norm. We continue with the details. We define two wedges with apex at the origin: Let $W^L=\{(x,y)\in \mathbb{R}^2: 0<y<x\}$ and $W^T=\{(x,y)\in \mathbb{R}^2: 0<x<y\}$; see Fig.~\ref{fig:wedges}(b--c). The Minkowski sums $W_i^L:=s_i+W_L$ and $W_i^T:=s_i+W^T$ are the translates of these wedges with apex at $s_i$. Let $r_i$ be a point of minimum $x$-coordinate in $W_i^L\cap (S\cup \partial U)$; and let $t_i$ be a point of minimum $y$-coordinate in $W_i^T\cap (S\cup \partial U)$. Then the blocker of $q_i^1$ is either $r_i$ or $t_i$, whichever is closer to $s_i$ in $L_\infty$ norm.

For every $i=1,\ldots, n$, we find points $r_i$ and $t_i$, independently. Consider the points $r_i\in W_i^R$, for $i=1,\ldots, n$ (the case of the points $t_i\in W_t^T$ is analogous). We use a data structure originally developed for computing $\Theta$-graphs in the context of geometric spanners by Narasimhan and Smid~\cite[Section~4.1.2]{NS07}. They developed the following dynamic data structure for $n$ points in the plane:
\begin{lemma}\label{lem:ns} \cite[Lemma~4.1.9]{NS07}.
Let $H$ be a nonvertical line through the origin. There is a data structure that maintains a set $P$ of $n$ points in the plane and supports
the following queries: (i) \textsc{MinBelow}$(p)$: Given a query point $p\in P$, compute a point with the minimum $x$-coordinate among all points in $P$ that are below $p+H$; (ii) insert a point into $P$; (iii) delete a point from $P$. The data structure has $O(n)$ space, $O(n\log n)$ preprocessing time, and $O(\log n)$ query time.
\end{lemma}

\begin{corollary}\label{cor:ns}
Given a  point set $S=\{s_i: i=1\ldots , n\}\subset U$, the points $r_i$ and $t_i$ ($i=1,\ldots , n$) can be computed in $O(n\log n)$ time.
Consequently, the squares $q_i^1$ can also be computed in $O(n\log n)$ time.
\end{corollary}
\begin{proof}
Assume that $S$ is sorted in decreasing order by their $y$-coordinates. We use the data structure in Lemma~\ref{lem:ns} with the line $H:y=x$ as follows. Initially $P=\emptyset$. For $i=1,\ldots, n$, we insert $s_i$ into $P$. If \textsc{MinBelow}$(s_i)$ returns a point in $P$, then let this be $r_i$, otherwise let $r_i$ be the point in the right side of $U$ that has the same $y$-coordinate as $s_i$. Since $P$ contains all points in $S$ whose $y$-coordinates are greater or equal to that of $s_i$, if wedge $W_i^R$ contains any anchor, then \textsc{MinBelow}$(s_i)$ returns one with the minimum $x$-coordinate. This shows that $r_i$ is computed correctly for all $i=1,\ldots, n$.

The points $t_i$ ($i=1,\ldots , n$) can be computed analogously in $O(n\log n)$ time. In $O(1)$ additional time
for each $i=1,\ldots ,n$, we can compare $r_i$ and $t_i$, find the blocker of $q_i^1$, and determine the maximal anchored square $q_i^1$.
\end{proof}

\begin{proof}[Proof of Theorem~\ref{thm:algorithm}]
By a repeated application of Corollary~\ref{cor:ns}, we can compute all $4n$ anchored squares $q_i^j$ $(i=1,\ldots, n$; $j=1,\ldots ,4$). As noted above, a sweep-line algorithm can compute the union $R(S)=\bigcup_{i=1}^n \bigcup_{j=1}^4 q_i^j$ in $O(n \log n)$ time. This completes the proof.
\end{proof}

\section{NP-Hardness of Maximum-Area Anchored Square Packings}
\label{sec:hardness}

We now prove that the maximum-area anchored square packing problem is NP-complete.
We define the decision version of the problem as follows. Instead of the unit square $[0,1]^2$,
we use the square $U=[0,W]^2$, for some integer $W>0$. For a finite set $S\subset [0,W]^2$ of anchors with integer coordinates, we ask whether there is an anchored square packing of area $W^2$.

We prove NP-hardness by a reduction from \textsc{Planar-Monotone-3SAT} (described below). For every instance of \textsc{Planar-Monotone-3SAT}, we construct an instance $S\subset [0,W]^2$. We say that an anchored empty square is \emph{forced} if every packing of area $W^2$ contains it. An anchor in $S$ is \emph{forced} if it is the anchor of a forced square; otherwise it is \emph{free}. A forced square $A$ and its anchor $s\in S$ form a \emph{forced pair} $(A,s)$. We construct an instance in which most of the anchors are forced, and a small number of anchors encode the truth value of the variables in a 3SAT instance.

To prove that the two instances are equivalent, we shall argue that a set of squares and anchors are forced. In an intermediate step, we assume that $\mathcal{F}=\{(A_i,s_i):i=1\ldots, f\}$ is a set of forced square-anchor pairs, and we would like to show that another square-anchor pair $(A,s)$ is also forced. Let $P=U\setminus \bigcup_{i=1}^f A_i$ be the complement of the forced squares in $\mathcal{F}$. By construction, $P$ is an orthogonal polygon (possibly with holes). An anchor $s$ is \emph{undecided} if there is no forced pair $(A,s)$ anchored at $s$ in $\mathcal{F}$ (i.e., $s$ is either free or its forced pair is not in $\mathcal{F}$).

We show (Lemma~\ref{lem:forced}) that the following two properties each imply that the pair $(A,s)$ is forced (given that all pairs in $\mathcal{F}$ are forced).
We define both properties for the orientations shown in Fig.~\ref{fig:forced},
but they generalize to all other orientations obtained through the symmetry group of $U$.
Let $s\in S$, and let $A$ be a maximal empty square anchored at $s$ such that ${\rm int}(A)\subset P$.
Without loss of generality, assume that $s$ is the upper-right corner of $A$.

\begin{figure}[h]
	\centering
	\includegraphics[width=0.7\linewidth]{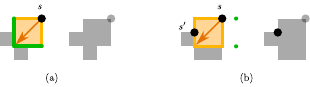}
	\caption{Identifying forced points. (a) Property~\ref{case1}. (b) Property~\ref{case2}.}
	\label{fig:forced}
\end{figure}

\begin{enumerate}
\item \label{case1}
The lower-left corner of $A$ is a convex vertex of $P$ and there is no undecided anchor in the closure of the bottom and left edges of $A$.
\item \label{case2}
The lower-left corner of $A$ is a convex vertex of $P$, the side length of $A$ is greater than 1, and the bottom edge of $A$ is contained in $\partial P$. There is a unique undecided anchor $s'\in \partial A$ located one unit above the lower-left corner of $A$. There is no undecided anchor one unit to the right of $s$ or to the right of the lower-right corner of $A$.
\end{enumerate}

\begin{lemma}\label{lem:forced}
Given a set of forced pairs $\mathcal{F}=\{(A_i,s_i):i=1\ldots, f\}$,
if a pair $(A,s)$ has properties~\ref{case1} or~\ref{case2}, then $(A,s)$ is a forced pair.
\end{lemma}
\begin{proof}
Suppose, to the contrary, that there is a anchored square packing $Q$ of area $W^2$ that does not use the square $A$ anchored at $s$.
Let $B\subseteq A$ be the unit square incident to the lower-left corner of $A$. Since all anchors have integer coordinates, every empty square containing $B$ is also contained in $A$.
If a pair $(A,s)$ has property~\ref{case1}, apart from $A$, no such empty square has a corner at an undecided anchor and, hence, $B$ cannot be covered.
If a pair $(A,s)$ has property~\ref{case2}, $B$ must be covered by a square anchored at its upper-left corner, which is undecided by hypothesis.
Hence, $B\in Q$.
Let $B'$ be a unit square with integer coordinates to the right of $B$.
Then the maximal empty square in $P\setminus B$ containing $B'$ satisfies property~\ref{case1},
but that there is no point at its upper-right corner.
In this case, $B'$ is not covered.
\end{proof}

\begin{theorem}\label{thm:hardness}
It is NP-hard to compute the maximum area anchored square packing of a given set $S$
of $n$ anchors with integer coordinates in a square $U=[0,W]^2$.
\end{theorem}
\begin{proof}
We reduce from \textsc{Planar-Monotone-3SAT} which is NP-complete~\cite{BK12}.
An instance of such problem consists of a boolean formula $\Phi$ in 3CNF with $n$ variables $\{x_1,\ldots,x_n\}$ and $m$ clauses,
and a planar rectilinear drawing of the a bipartite graph of $\Phi$.
The drawing given by an \textsc{Planar-Monotone-3SAT} instance represents variables and clauses by rectangles,
and edges by vertical line segments. It has the additional property that the rectangles of variables
(and only variables) intersect the line $y=0$ and the rectangles of clauses lies in the upper (resp., lower) half-plane contain only positive (resp., negative) literals.
%
%
A literal is called \emph{negative} if it is the negation of a variable, and \emph{positive} otherwise.
We need to decide whether we can satisfy all $m$ clauses, each of which is a disjunction of three literals.
	
	\begin{figure}[h]
		\centering
		\includegraphics[width=\textwidth]{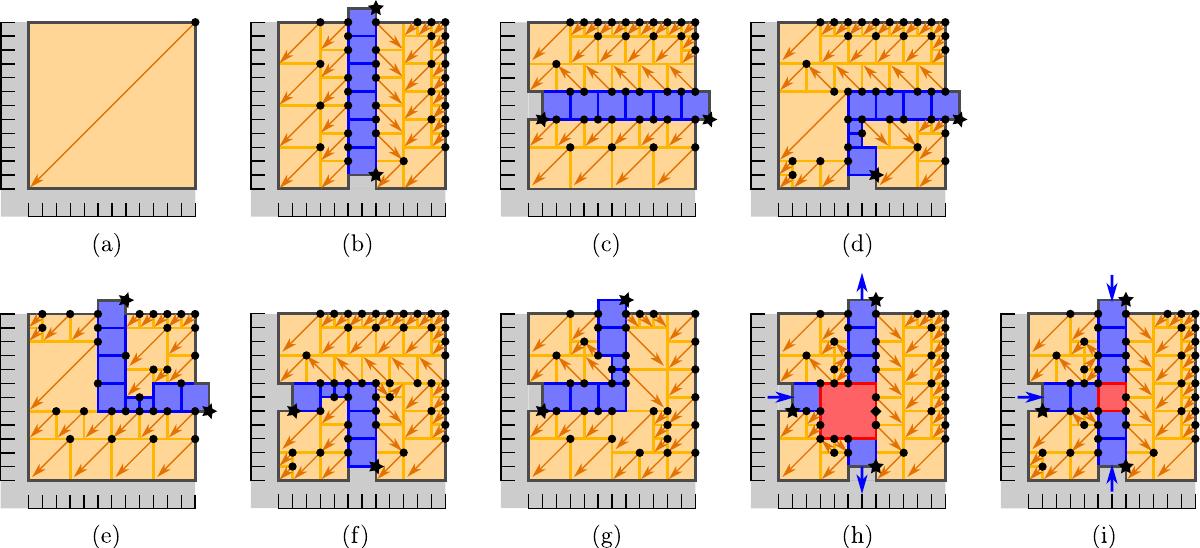}
		\caption{Gadgets.
     (a) is a \emph{filler} gadget, (b--c) are \emph{wire} gadgets, (d--g) are \emph{turn} gadgets, (h) is a split gadget, and (i) is a \emph{clause} gadget.
     The rhombus in (h) represents 2 anchors placed at the same position.}
		\label{fig:gadgets}
	\end{figure}

For a given instance of \textsc{Planar-Monotone-3SAT}, we construct an instance $S\subset [0,W]^2$ of the maximal area anchored square packing problem, and then show that the two instances are equivalent.
We first modify the rectilinear graph of the \textsc{Planar-Monotone-3SAT} instance in the following way.
Replace each rectangle by a cycle along its boundary and denote by $G$ the resulting geometric graph.
Delete the left, right, and the top (resp., bottom) edges of the rectangles representing positive (resp., negative) clauses. Each clause is now represented by a horizontal segment (a path of length 2 in $G$).
We designate the middle vertex of this path, which has degree 3, as a \emph{clause vertex}.
For each cycle in $G$ that represents a variable, delete the right vertical edges,
and designate the left vertical edge as a \emph{variable edge}.
All remaining edges in $G$ called \emph{wires} and all remaining vertices of degree 3 are called \emph{split vertices}. We orient the wires such that they form directed paths from the variable edges to clause vertices.
Assume that the feature size of the resulting rectilinear graph is 1 and the side length of a minimum enclosing axis-aligned square is $k$. We set $W=48k+48$ and let $U=[0,W]^2$.	
Scale up the drawing by a factor of $48$ and place it in $U$ so that
every vertex is at distance at least 24 from $\partial U$.

We tile $U$ with orthogonal polygons. Every tile is congruent to one of the tiles shown in Fig.~\ref{fig:gadgets}. We call these tiles \emph{gadgets}:
(a) is a \emph{filler} gadget, (b--c) are \emph{wire} gadgets, (d--g) are \emph{turn} gadgets, (h) is a split gadget, and (i) is a \emph{clause} gadget. The filer gadget is a $12\times 12$ square,
all other tiles are constructed from a $12\times 12$ square by possibly adding or deleting
$1\times 2$ rectangular features in two or three side of the squares. In a tiling of $[0,W]^2$,
each such feature matches a feature of an adjacent tile.
Choose a tile for each variable that contains part of the variable edge and add the anchors shown in Fig.~\ref{fig:gadgets}(b) (only the star contained in the $12\times 12$ square is added).
Do the same for split and clauses using the tiles shown in Fig.~\ref{fig:gadgets}(h) and (i), respectively.
Connect the gadgets as they are connected in the original drawing using wires and turns.
The directions of the wires attached to split and clause gadgets are indicated by blue arrows in Fig.~\ref{fig:gadgets}.
For all remaining tiles, we use filler gadgets with one anchor.
This completes the description of the instance $S\subset [0,W]^2$.
	
	\begin{figure}[h]
		\centering
		\includegraphics[width=.2\textwidth]{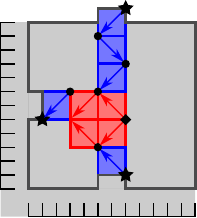}
		\caption{Square packing for the split gadget connected to negative wires.}
		\label{fig:false-split}
	\end{figure}
	
We now prove that a \textsc{Planar-Monotone-3SAT} instance admits a positive solution if and only if the corresponding point set $S\subset [0,W]^2$ admits an anchored square packing of area $W^2$.

Assume that the \textsc{Planar-Monotone-3SAT} instance admits a positive solution. We show that the corresponding instance $S\subset [0,W]^2$ admits an anchored square packing of area $W^2$.
Choose every orange square in each gadget assigning its anchor as the origin of the arrow contained in it (as shown in Fig.~\ref{fig:gadgets}).
For each of the $n$ wire gadgets placed on variable edges, if the corresponding variable is assigned \texttt{true} (resp., \texttt{false}) add all blue squares assigning its anchor as the only point on its bottom (resp., top) edge.
Every connected component formed by blue squares represents a path of wires.
We say that a wire is \emph{positive} if it is in the upper half of $U$ and its corresponding variable is assigned \texttt{true}, or if it is in the lower half of $U$ and its corresponding variable is assigned \texttt{false}. A wire is \emph{negative} otherwise. If a wire is positive (resp., negative), assign the corner that is behind (resp., ahead of) the blue square as its anchor, considering the direction of the wire.
For split gadgets connected to positive wires assign the top-left corner of the red square as its anchor.
For negative wires, the red square is reached by four equal squares as shown in Fig.~\ref{fig:false-split}.
Since there exists at least one positive wire connected to a clause vertex, there will be at least one point in a corner of the red square in the clause gadget that has not yet been assigned a square. We complete the square packing by adding such a square with a corresponding anchor. Since the anchored squares cover all gadgets, the overall area of the square packing is $W^2$.

Assume that the anchored square packing instance $S\subset [0,W]^2$ admits a positive solution (of area $W^2$).
Recall that $[0,W]^2$ is tiled with gadgets. Sort them in lexicographic order by the coordinates of their lower-left corners (i.e., the first gadget is incident to the origin). We use Lemma~\ref{lem:forced} to prove the following property for each gadget:
\begin{quote}
	\textbf{Property (i).} If the left and bottom boundaries are part of the perimeter of a forced polygon $P$ and contain no free anchor relative to $P$ except for the points shown by a star, then (i.a) every orange and blue square shown in Fig.~\ref{fig:gadgets} in the corresponding gadget is forced; and (i.b) if $P'$ is the union of $P$ and the orange and blue squares inside the gadget, then there is no free anchor relative to $P'$ on the boundary of the gadget except for points shown by a star.
\end{quote}

\begin{figure}[h]
	\centering
	\includegraphics[width=\textwidth]{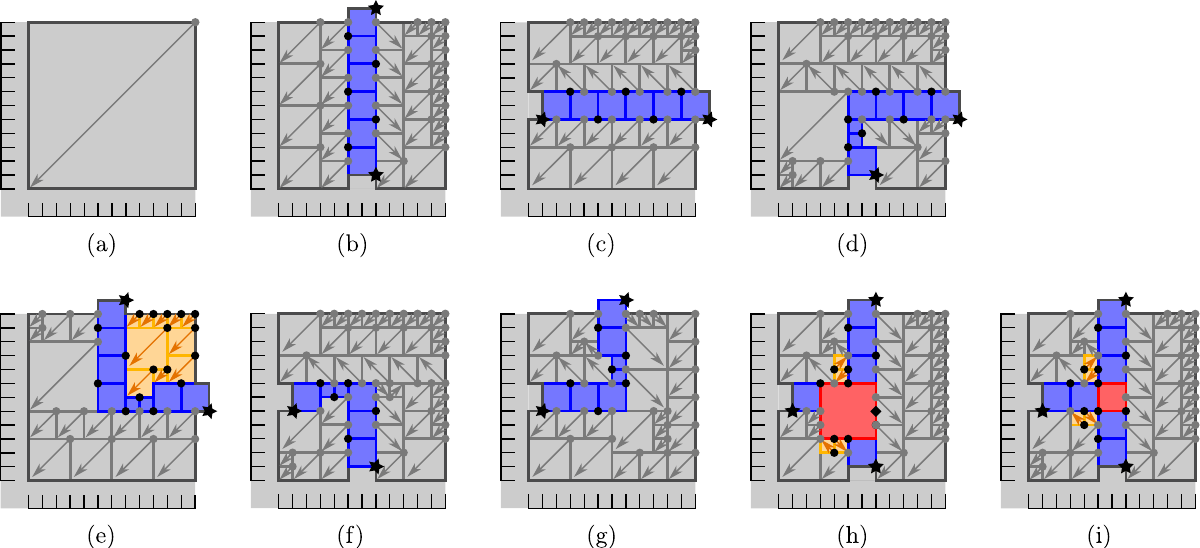}
	\caption{Forced squares and anchors.}
	\label{fig:gadgets-forced}
\end{figure}

Initially, in every gadget, we can determine at least one pair of a forced square and a corresponding forced anchor using Lemma~\ref{lem:forced}. Fig.~\ref{fig:gadgets-forced} shows the result of recursively adding a forced square into the forced polygon, and applying Lemma~\ref{lem:forced} to another pair until there are no more forced pairs in the gadget.
We now show that all blue squares are forced.
In each case, we can take the lower-left blue square and conclude that if the square packing covers it entirely, then it is covered by a square anchored at one of its corners.	After we add this square to the forced polygon $P$, the same argument holds for every lower-left blue square not in $P$. Consequently, all blue squares are forced.
The remaining orange squares are forced by recursively applying Lemma~\ref{lem:forced}.
Then, if property (i) is satisfied, every gadget satisfies (i.a) and (i.b).
Property (i) is trivially satisfied for the lower-left gadget and inductively satisfied by assuming that all gadgets to the left and below satisfy (i.a) and (i.b).

We now show how to convert a square packing of area $W^2$ into a solution of the \textsc{Planar-Monotone-3SAT} instance. Wire gadgets have two points indicates by a star: one that is ahead and one behind using the direction of the wire (recall that the direction points from the variable edge to the clause).
A wire gadget that does not use the star that is ahead in its direction as an anchor for one of the squares contained in it is called \emph{positive}.
A wire gadget is called \emph{negative} otherwise.
For all wire gadgets satisfying (i.a) and (i.b), if a star is not used as an anchor for a square in the gadget, then the other point marked by a star must be used as anchor in this gadget.
This implies that, for a pair of adjacent wire gadgets, if the one ahead in the wire direction is positive, so is the other gadget.
Now assume that one of the outputs of the split gadget satisfying (i.a) and (i.b) is connected to a positive gadget.
A point in the middle of an edge of the red square in Fig.~\ref{fig:gadgets-forced}(h) must be used as an anchor of a blue square.
Then, the only way to cover all the red area is to use a single square anchored at its upper-left corner.
Therefore, the wire connected to the input of the split gadget must also be positive.
Finally, assume that the red square in a clause gadget that satisfies (i.a) and (i.b) (see Fig.~\ref{fig:gadgets-forced}(i)) is covered.
Then, it must be anchored at one of its corners.
If it is anchored at the upper-left (resp., bottom-left, bottom-right) corner, then the star at the top (resp., left, bottom) of the gadget is used as an anchor of a blue square in this gadget.
Therefore, it must be adjacent to a positive wire.
Combining all arguments, we set a variable \texttt{true} if its first wire gadget (that was placed on the variable edge) is positive and \texttt{false} otherwise, and then this assignment will satisfy the boolean formula of the \textsc{Planar-Monotone-3SAT} instance.
\end{proof}

\section{Open Problems}
\label{sec:open}

We have shown that at least half of the area of the unit square $U=[0,1]^2$ can be reached by empty squares anchored at $S$ for any finite set $S\subset U$, and this bound is the best possible. We have also given the first NP-hardness proof for a packing problem over geometric objects of arbitrary sizes. Our results raise several intriguing open problems.
Does our result generalize to higher dimensions, that is, is there a lower bound for the maximal volume covered by empty hypercubes anchored at a finite set of points in $[0,1]^d$ for $d>2$? Axis-aligned squares are balls in $L_\infty$-norm: Over all finite sets $S$ of anchors in a unit-diameter ball $U$ in $L_p$-norm, $p\geq 1$, what is the maximum area of a packing of $L_p$-balls that each contain an anchor?
Is there a polynomial-time algorithm for computing the minimum area \emph{lower-left anchored} square packing for a given set $S$ of $n$ points in the unit square $[0,1]^2$?
Is it NP-hard to compute the maximum area anchored \emph{rectangle} packing of a given set $S\subset [0,1]^2$?
For the last two problems, simple greedy strategies achieve constant-factor approximations~\cite{DT15},
and a QPTAS is available for rectangles and a PTAS for squares~\cite{BDT17}.

\bibliographystyle{plainurl}
\bibliography{reach}

\end{document}